\documentclass[11pt]{article}

\usepackage{a4wide}
\usepackage{authblk}
\usepackage{amsmath,amssymb}
\usepackage{amsthm}
\usepackage[dvipdfmx]{graphicx}

\usepackage{here}
\usepackage{comment}

\pdfoutput=1

\usepackage{thmtools}

\newtheorem{theorem}{Theorem}
\newtheorem{lemma}{Lemma}
\newtheorem{corollary}{Corollary}

\theoremstyle{definition}
\newtheorem{definition}{Definition}

\newtheorem{example}{Example}

\newtheorem{fact}{Fact}

\usepackage{array}
\usepackage{ascmac}
\usepackage{booktabs}
\usepackage{wrapfig}
\usepackage{url}
\usepackage{here}

\usepackage[final,mode=multiuser]{fixme}
\usepackage{xcolor}

\fxuseenvlayout{colorsig}
\fxusetargetlayout{color}
\FXRegisterAuthor{hf}{ahf}{HF}
\FXRegisterAuthor{si}{asi}{SI}
\fxsetface{env}{}

\newcommand{\Prefix}{\mathsf{Prefix}}
\newcommand{\Substr}{\mathsf{Substr}}
\newcommand{\Suffix}{\mathsf{Suffix}}

\newcommand{\CDAWG}{\mathsf{CDAWG}}
\newcommand{\grammar}{\mathsf{G_{CDAWG}}}

\newcommand{\EndPos}{\mathsf{EndPos}}
\newcommand{\BegPos}{\mathsf{BegPos}}

\newcommand{\EqrL}{\equiv^{\mathrm{L}}}
\newcommand{\EqrR}{\equiv^{\mathrm{R}}}

\newcommand{\LeftM}{\mathsf{LeftM}}
\newcommand{\RightM}{\mathsf{RightM}}
\newcommand{\M}{\mathsf{M}}
\newcommand{\MR}{\mathsf{MR}}
\newcommand{\D}{\mathsf{d}}

\newcommand{\G}{\mathsf{C}}

\newcommand{\size}{\mathsf{e}}
\newcommand{\Gsize}{\mathsf{g}}
\newcommand{\vone}{\mathrm{v}^{(1)}}
\newcommand{\I}{\mathsf{I}}

\newcommand{\remedge}{\mathsf{rem\_e}}

\newcommand{\AS}{\mathsf{AS}}

\newcommand{\Q}{\mathsf{Q}}

\newcommand{\N}{\mathsf{N}}
\newcommand{\None}{\mathsf{N}_1}
\newcommand{\Ntwo}{\mathsf{N}_2}
\newcommand{\Nthree}{\mathsf{N}_3}
\newcommand{\Nadd}{\mathsf{N}_{3\mathrm{B}}}
\newcommand{\Nbase}{\mathsf{N}_{3\mathrm{A}}}

\newcommand{\uc}{\mathcal{J}}
\newcommand{\wc}{\mathcal{L}}

\newcommand{\VO}{\mathsf{V}_{\mathrm{1}}}
\newcommand{\VA}{\mathsf{V}_{\mathrm{2(a)}}}
\newcommand{\VB}{\mathsf{V}_{\mathrm{2(b)}}}
\newcommand{\VC}{\mathsf{V}_{\mathrm{2(c)}}}

\newcommand{\Vbt}{\mathsf{V}_{\mathrm{2(c)-1}}}
\newcommand{\Val}{\mathsf{V}_{\mathrm{2(c)-2}}}
\newcommand{\Valy}{\mathsf{V}^{(y)}_{\mathrm{2(c)-2}}}

\newcommand{\Valz}{\mathsf{V}^{(z)}_{\mathrm{2(c)-2}}}

\newcommand{\Sy}{\mathit{D}}
\newcommand{\Xy}{\mathit{X}}
\newcommand{\Snoty}{\hat{\mathit{D}}}

\newcommand{\zl}{z_l(x)}
\newcommand{\zly}{z_l(x_y)}

\newcommand{\rrep}{\mathsf{rexp}}

\newcommand{\Ndelta}{\delta_\mathsf{N}}
\newcommand{\Vdelta}{\delta_\mathsf{V}}

\begin{document}

\title{On the sensitivity of CDAWG-grammars}

\author[1]{Hiroto~Fujimaru}
\author[2]{Shunsuke~Inenaga}

\affil[1]{Department of Information Science and Technology
{\tt fujimaru.hiroto.134@s.kyushu-u.ac.jp}}

\affil[2]{Department of Informatics, Kyushu University, Japan
{\tt inenaga.shunsuke.380@m.kyushu-u.ac.jp}}

\date{}
\maketitle

\begin{abstract}
The \emph{compact directed acyclic word graph} (\emph{CDAWG})
    [Blumer et al. 1987] of a string is the minimal compact automaton
    that recognizes all the suffixes of the string.
    CDAWGs can be used for various string tasks
    including text pattern searching, data compression, and pattern discovery.
    The \emph{CDAWG-grammar} [Belazzougui \& Cunial 2017] is a grammar-based text compression based on the CDAWG,
    which allows for representing the CDAWG in $O(\size)$ space
    without storing the string, where $\size$ denotes the number of CDAWG edges.
    Let $\Gsize$ be the size of the CDAWG-grammar for the input string $T$.
    We show that the worst-case additive sensitivity of the CDAWG-grammar
    is lower bounded by $3\Gsize-21$ and is upper bounded by $8 \Gsize + 4$.
\end{abstract}


\section{Introduction}
\label{sec:intro}

The \emph{compact directed acyclic word graph} (\emph{CDAWG})~\cite{Blumer1987} of a string is the minimal compact automaton
that recognizes all the suffixes of the string.
CDAWGs are known to be useful for various string tasks
including text pattern searching~\cite{Crochemore1997,Inenaga2005}, data compression~\cite{BelazzouguiC17_CPM,BelazzouguiC17_SPIRE,Takagi2017}, and pattern discovery~\cite{Takeda2000}.

CDAWGs can be obtained by merging isomorphic subtrees of suffix trees~\cite{Weiner1973}.
Due to this nature, as in the case of suffix trees of which each edge string label $x$ is represented by a pair $(i,j)$ of positions satisfying $T[i..j] = x$,
classical implementations of CDAWGs also require
$\Theta(n)$ space to explicitly store a copy of the string $T$ being indexed, where $n = |T|$.
This significantly limits the size of strings on which the CDAWGs are built.
Still, the number $\size$ of the edges in the CDAWG tends to be small,
for some highly repetitive strings,
in particular, $\size = O(\log n)$ holds for the standard Sturmian words and Thue-Morse words~\cite{Rytter06,BaturoPR09,RadoszewskiR12}.

The task of representing the CDAWG with $O(\size)$ space was first achieved by Belazzougui and Cunial~\cite{BelazzouguiC17_CPM,BelazzouguiC17_SPIRE}, via a connection from CDAWGs to \emph{grammar-based compressions}.
They observed that the reversed CDAWG induces a context-free grammar (CFG)
that only generates the input string $T$, which is hereafter named as the \emph{CDAWG-grammar}.
By augmenting the CDAWG-grammar with a constant-time level ancestor data structure,
the CDAWG can be stored in $O(\size)$ space while retaining the ability to perform optimal $O(m+occ)$-time pattern matching, where $m$ is the pattern length and $occ$ is the number of occurrences to report~\cite{BelazzouguiC17_CPM,BelazzouguiC17_SPIRE}.
While the CDAWG-grammar can easily be obtained via the graph topology of the corresponding CDAWG,
it is also possible to build the CDAWG-grammar \emph{directly} from suffix-array-based structures~\cite{AlanJ2023}, without the need of explicitly building the CDAWG.




The focus of this paper is to analyze the advantage of CDAWG-grammar
in terms of \emph{sensitivity}:
The sensitivity of a repetitiveness measure $\mathsf{c}$ asks how much the measure size increases when a single-character edit operation is performed on the input string~\cite{AkagiFI2023}, 
and thus the sensitivity allows one to evaluate the robustness of the measure/compressor against errors/edits.
%
Recent works have revealed the sensitivity of CDAWG size $\size$
in the case of left-end edit operations~\cite{FujimaruNI25}
as well as in the general case~\cite{HamaiFI2025}.
As for the CDAWG-grammar size $\Gsize$,
it is known that $\Gsize  = \size - \vone$ holds~\cite{BelazzouguiC17_CPM,BelazzouguiC17_SPIRE}, where $\vone$ denotes the number of CDAWG nodes of in-degree one.
To obtain non-trivial bounds for the worst-case additive sensitivity
of the CDAWG-grammar size $\Gsize$,
we present new properties of $\vone$ which were not previously well understood.

In this paper, we first provide a lower bound of $3\Gsize - 21$.
After presenting a straightforward upper bound of $15 \Gsize + 4$,
we give a non-trivial tighter upper bound of $8 \Gsize + 4$
based on combinatorial properties of CDAWGs and CDAWG-grammars.
Fig.~\ref{fig:roadmap} shows a roadmap for proving our upper bound $8 \Gsize + 4$.
Our bounds for the worst-case additive sensitivity
of the CDAWG-grammar are optimal up to a constant factor.

\begin{figure}[H]
    \centering
    \includegraphics[width=1.0\textwidth]{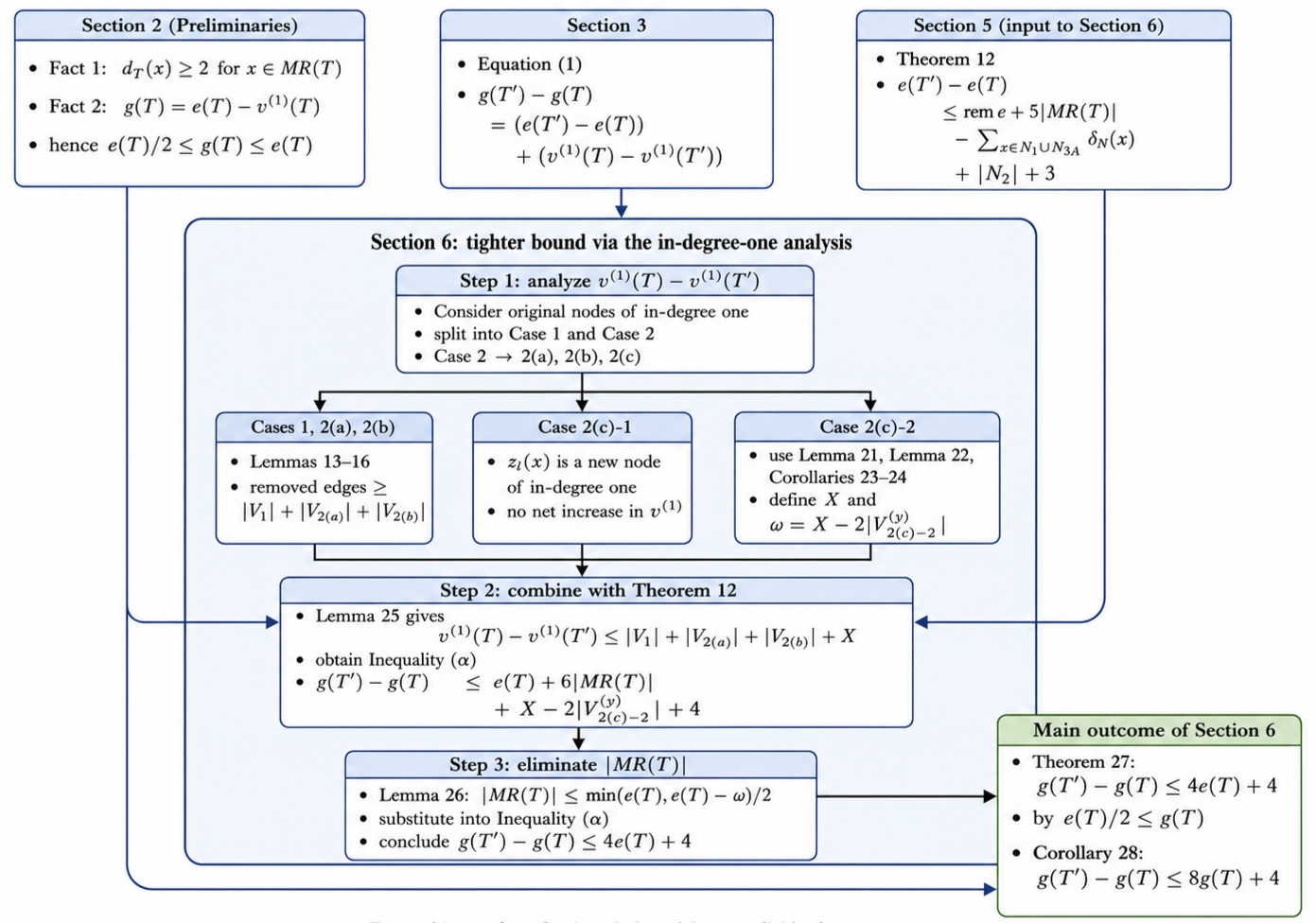}
    \caption{The roadmap for the $8 \Gsize + 4$ upper bound
    on the worst-case additive sensitivity of the CDAWG-grammar size $\Gsize$.}
    \label{fig:roadmap}
\end{figure}





\section{Preliminaries}

\noindent \textbf{Strings.}
Let $\Sigma$ be an \emph{alphabet} of size $\sigma$.
An element of $\Sigma^*$ is called a \emph{string}.
For a string $T \in \Sigma^*$, the length of $T$ is denoted by $|T|$.
The \emph{empty string}, denoted by $\varepsilon$, is the string of length $0$.
Let $\Sigma^+ = \Sigma^* \setminus \{\varepsilon\}$.
If $T = uvw$, then $u$, $v$, and $w$ are called a \emph{prefix}, \emph{substring},
and \emph{suffix} of $T$, respectively.
The sets of prefixes, substrings, and suffixes of string $T$ are denoted by
$\Prefix(T)$, $\Substr(T)$, and $\Suffix(T)$, respectively.
For a string $T$ of length $n$, $T[i]$ denotes the $i$th character of $T$
for $1 \leq i \leq n$,
and $T[i..j] = T[i] \cdots T[j]$ denotes the substring of $T$ that begins at position $i$ and ends at position $j$ on $T$ for $1 \leq i \leq j \leq n$.
\hfnote*{modified}{%
 For $0 \leq k \leq |T|$, let $\varepsilon = T[k+1..k]$.
}

\noindent \textbf{Maximal substrings.}
For two strings $u$ and $T$,
let $\BegPos(u, T) = \{i \mid T[i..i+|u|-1] = u\}$ and
$\EndPos(u, T) = \{i \mid T[i-|u|+1..i] = u\}$ denote
the sets of beginning and ending positions of $u$ in $T$,
respectively.
For any substrings $u, v \in \Substr(T)$ of a string $T$,
define $u \EqrL_T v \Leftrightarrow \EndPos(u, T) = \EndPos(v, T)$,
and $u \EqrR_T v \Leftrightarrow \BegPos(u, T) = \BegPos(v, T)$.
A substring $x \in \Substr(T)$ of $T$ is said to be
\emph{left-maximal} in $T$ if
$ax, bx \in \Substr(T)$ for some two distinct characters $a,b \in \Sigma$, or
$x \in \Prefix(T)$.
Similarly, $x$ is said to be \emph{right-maximal} in $T$ if
$xa, xb \in \Substr(T)$ for some two distinct characters $a,b \in \Sigma$, or
$x \in \Suffix(T)$.
Let $\LeftM(T)$ and $\RightM(T)$ denote the sets of left-maximal
and right-maximal substrings in $T$, respectively.
It is known that $x \in \LeftM(T)$ (resp. $x \in \RightM(T)$) iff
$x$ is the longest member of the equivalence class under $\EqrL_T$
(resp. under $\EqrR_T$).
Let $\mathsf{M}(T) = \LeftM(T) \cap \RightM(T)$.
Any element of $\mathsf{M}(T)$ is said to be \emph{maximal} in $T$.

A string $x$ is said to be a repeat in a string $T$
if $|\BegPos(x, T)| = |\EndPos(x, T)|$ $\geq 2$.
Let $\MR(T) = \M(T) \setminus \{T\}$ denote the set of \emph{maximal repeats} in $T$.


For any substring $S$ of a string $T$,
we define its \emph{right-representative}
as $\rrep_T(S)$ $= S \beta$,
where $\beta \in \Sigma^*$ is the shortest string
such that $S\beta$ is right-maximal in $T$.

\noindent \textbf{CDAWGs.}
The \emph{suffix tree}~\cite{Weiner1973}
of a string $T$ is a rooted tree such that
each edge is labeled by a non-empty substring of $T$
and the path from the root to each node $v$ represents
a right-maximal substring in $\RightM(T)$.
By assuming a unique end-marker $\$$ at the right-end of $T$,
every internal node of the suffix tree of $T$ has at least two children,
and there is a one-to-one correspondence between
the leaves and the suffixes of $T$.
The \emph{compact directed acyclic word graph} (\emph{CDAWG})~\cite{Blumer1987}
of a string $T$,
denoted $\CDAWG(T) = (\mathsf{V}_T, \mathsf{E}_T)$,
is an edge-labeled smallest DAG that is obtained by
merging isomorphic subtrees of the suffix tree for $T$.
See the left diagram of Fig.~\ref{fig:CDAWG_grammar} for a concrete example of CDAWGs.
\begin{figure}[t]
  \centering
  \includegraphics[keepaspectratio,scale=0.39]{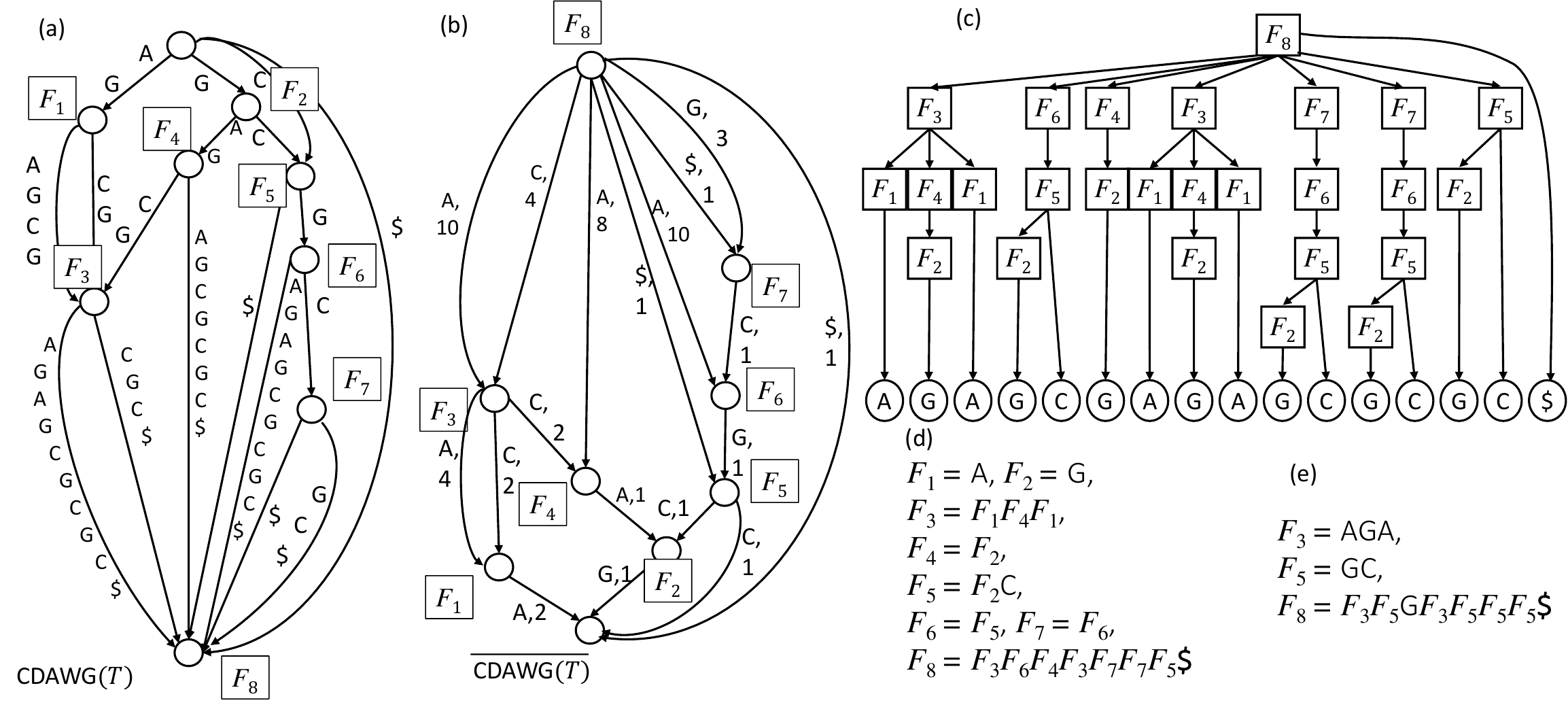}
  \caption{(a) $\CDAWG(T)$ of string $T=\mathtt{AGAGCGAGAGCGCGC}\$$ for which $\M(T) = \{\varepsilon, \mathtt{G}, \mathtt{GC}, \mathtt{AG}, \mathtt{GAG}, \mathtt{GCG}, \mathtt{GCGC}, \mathtt{AGAGCG}, T\}$. The number of right-extensions of $\CDAWG(T)$ is the number $\size(T)$ of edges, which is 18 in this example.
   (b) The reversed DAG $\overline{\CDAWG(T)}$, where each edge is labeled by the initial character and the length of the edge's original label.
   (c) The derivation tree $\mathcal{T}(T)$ obtained by unfolding $\overline{\CDAWG(T)}$.
   (d) The grammar rules obtained from $\mathcal{T}(T)$.
   (e) The resulting grammar $\grammar(T)$ without redundant rules. This grammar size is 13.}
  \label{fig:CDAWG_grammar}
\end{figure}
It is known that there is a one-to-one correspondence between
the nodes of $\CDAWG(T)$ and the maximal substrings in $\M(T)$,
such that the longest string represented by each node of $\CDAWG(T)$
is a maximal substring.
In what follows, we identify the nodes of $\CDAWG(T)$
with the maximal substrings in $\M(T)$.
  In particular,
  \hfnote*{modified}{%
  we identify the non-sink nodes of $\CDAWG(T)$
  }%
  with the maximal repeats in $\MR(T)$.
For any $x \in \MR(T)$, $\D_{T}(x)$ denotes the out-degree of the node $x$ in $\CDAWG(T)$.
Similar to suffix trees, every internal node $x$ of the $\CDAWG(T)$ has at least two children if $T$ terminates with $\$$.
Hence, the following fact holds.
\begin{fact}
  \label{fact:outdegree}
  If $T$ terminates with $\$$, $\D_{T}(x) \geq 2$ for any $x \in \MR(T)$.
\end{fact}

The \emph{size} of $\CDAWG(T) = (\mathsf{V}_T, \mathsf{E}_T)$
for a string $T$ of length $n$
is the number $\size(T) = |\mathsf{E}_T|$ of edges in $\CDAWG(T)$,
which is also referred to as the number of right-extensions of maximal repeats in $T$.

\noindent \textbf{CDAWG-grammars.}
A \emph{context-free grammar} (CFG) $\mathcal{G}$ which generates only a string $T$ is called a grammar compression for $T$. 
The size of $\mathcal{G}$ is the total length of the right-hand sides of all the production rules in $\mathcal{G}$. 

Suppose that the input string $T$ terminates with a unique end-marker $\$$.
The \emph{CDAWG grammar}~\cite{BelazzouguiC17_CPM,BelazzouguiC17_SPIRE} of $T$, denoted $\grammar(T)$,
is a grammar compression for $T$ that is built on $\CDAWG(T)$ as follows (see also Fig.~\ref{fig:CDAWG_grammar}):
Let $\overline{\CDAWG(T)}$ denote the DAG obtained by reversing all edges of $\CDAWG(T)$.
A critical observation is that,
since $\CDAWG(T)$ is a DAG obtained by merging subtrees of the suffix tree for $T$,
the paths in $\CDAWG(T)$ have a one-to-one correspondence between the suffixes of $T$,
and so do the paths in $\overline{\CDAWG(T)}$.
We label each non-source node $v$ in $\CDAWG(T)$
\hfnote*{modified}{%
by a non-terminal $F_i$ for 
$1 \leq i \leq |\mathsf{V}_T|-1$ 
}%
in a depth-first manner,
and let $p_k$ denote the path of $\CDAWG(T)$ that represents the suffix $T[k..|T|]$ starting at position $k$ in $T$.
Then, we build a tree $\mathcal{T}(T)$ by inserting, for each $1 \leq k \leq |T|$,
the reversed path $\overline{p_k}$ of $\overline{\CDAWG(T)}$
as the path of $\mathcal{T}(T)$ from the root to the $k$th leaf
that corresponds to the $k$th character $T[k]$.
In other words,
$\mathcal{T}(T)$ is the tree obtained by unfolding $\overline{\CDAWG(T)}$ into a tree.
Observe that $\mathcal{T}(T)$ is the derivation tree of a grammar that generates (only) $T$,
and its DAG representation is $\overline{\CDAWG(T)}$.
This initial grammar may contain
some redundant rules of form $F_i \rightarrow F_j$ (e.g., $F_6 \rightarrow F_5$ in Diagram (c) of Fig.~\ref{fig:CDAWG_grammar}).
We remove all such redundant rules, and the resulting grammar is the CDAWG grammar $\grammar(T)$
for $T$.

Let $\Gsize(T)$ denote the size of $\grammar(T)$.
Let $F_i \rightarrow w \in (\Sigma \cup \{F_1, \ldots, F_{i-1}\})^+$
be a production in $\grammar(T)$.
The length of the right-hand side of the production in $\grammar(T)$
is equal to the number of outgoing edges of the node of $\overline{\CDAWG(T)}$ that corresponds to $F_i$.
This is equal to the number of in-edges of the node of $\CDAWG(T)$.
Thus, we have $\Gsize(T) \leq \size(T)$.
Recalling that the redundant rules have been removed in $\grammar(T)$,
we obtain the following equation with the number $\vone(T)$ of nodes of in-degree 1 in $\CDAWG(T)$.
\begin{fact}
  \label{fact:Gsize}
  $\Gsize(T) = \size(T) - \vone(T)$.
\end{fact}

\section{Sensitivity of CDAWG-grammar $\Gsize$ and naive bound}
\label{sec:naive}

We define the worst-case \emph{additive sensitivity} of the CDAWG-grammar
for any string of length $n$ that terminates with $\$$ as follows:
\begin{eqnarray}
  \AS(\Gsize, n) & = & \max_{T \in \Sigma^n, T' \in \mathcal{K}(T,1)} \{\Gsize(T') - \Gsize(T) \mid T[|T|] = T'[|T'|] = \$\} \nonumber \\
  & = & \max_{T \in \Sigma^n, T' \in \mathcal{K}(T,1)} \{\size(T') \! - \! \size(T) \! + \! \vone(T) \! - \! \vone(T') \mid T[|T|] \! = \! T'[|T'|] \! = \! \$\} \label{eqn:grammar_size}
\end{eqnarray}
where $\mathcal{K}(T,1)$ denotes the set of strings
of edit distance 1 from $T$.
The second equation is due to Fact~\ref{fact:Gsize}.


\begin{example}
  \label{ex:size_diff}
  Let $T = \mathtt{abcabab\$}$ and $T' = \mathtt{\underline{b}}\mathtt{abcabab\$}$, 
  where $\mathtt{\underline{b}}$ is prepended to $T$. 
  The CDAWG and CDAWG-grammar of $T$ and $T'$ are shown in Fig.~\ref{fig:size_diff}.
  The size difference of the CDAWG is $\size(T') - \size(T) = 5$, while that of the CDAWG-grammar is $\Gsize(T') - \Gsize(T) = 2$.
\end{example}
\begin{figure}[t!]
  \centering
  \includegraphics[keepaspectratio,scale=0.5]{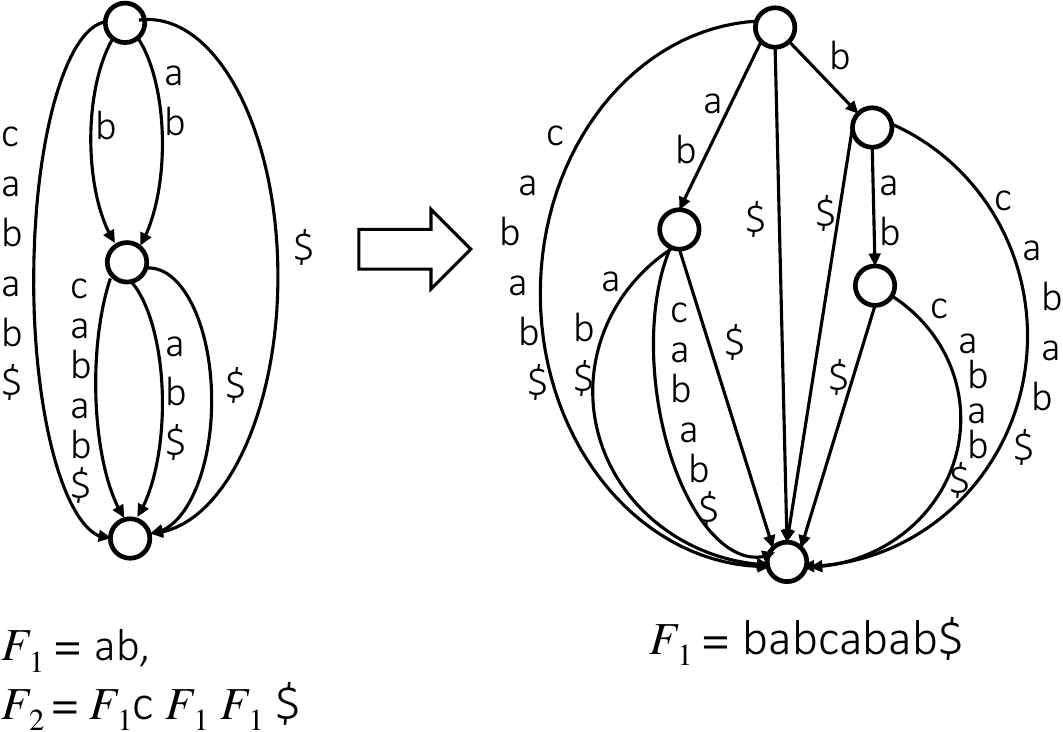}
  \caption{The CDAWG and CDAWG-grammar for the strings $T$ and $T'$ in Example~\ref{ex:size_diff}.}
  \label{fig:size_diff}
\end{figure}

It follows from Fact~\ref{fact:outdegree} that
$2\vone(T) \leq 2|\MR(T)| \leq \size(T)$.
Thus, by Fact~\ref{fact:Gsize} we get $\size(T)/2 \leq \Gsize(T) \leq \size(T)$.
Hamai et al.~\cite{HamaiFI2025} showed that $\size(T') \leq 8\size(T) + 4$ holds for any strings $T$ and $T' \in \mathcal{K}(T,1)$.
Using these inequalities, it immediately follows that
$\Gsize(T') \leq \size(T') \leq 8\size(T)+4 \leq 16\Gsize(T)+4$,
which implies $\AS(\Gsize, n) \leq 15\Gsize(T) + 4$.
In the following sections, we improve the constant factor of $\AS(\Gsize, n)$
from $15$ to $8$.

\section{Lower bound for the sensitivity of CDAWG-grammar}

Below, we present a non-trivial lower bound for $\AS(\Gsize, n)$.

\begin{theorem} \label{theo:lowerbound_new}
  There exists a family of strings $T$ and strings $T'\in K(T,1)$ such that
  \[
    \Gsize(T')-\Gsize(T)=3\Gsize(T)-21.
  \]
  Therefore $\AS(\Gsize, n) \geq 3\Gsize-O(1)$.
\end{theorem}

\begin{proof}
Let
\[
T=(ab)^m c (ab)^m c\$
\]
for an integer $m\ge 2$, where $a,b,c,\$\in\Sigma$ are pairwise distinct characters.
Let $T'$ be obtained from $T$ by substituting the first $c$ with $b$, namely
\[
T'=(ab)^m b (ab)^m c\$.
\]
Clearly, $T'\in K(T,1)$, and both $T$ and $T'$ terminate with the same end-marker $\$$.

We first compute $\Gsize(T)$. The maximal substrings of $T$, except for the sink $T$, are precisely
\[
\varepsilon,\quad (ab),(ab)^2,\ldots,(ab)^{m-1},\quad (ab)^m c.
\]
Indeed, for $1\le r\le m-1$, the string $(ab)^r$ is both left- and right-maximal: it has occurrences preceded by different left contexts (the beginning of $T$, $b$, or $c$) and occurrences followed by different right contexts $a$ and $c$. Also, $(ab)^m c$ is both left- and right-maximal, since its two occurrences are preceded by the beginning of $T$ and by $c$, and are followed by $a$ and $\$$, respectively. No other substring is both left- and right-maximal: a substring not containing $c$ is a factor of the alternating blocks $(ab)^m$, and only the even aligned factors $(ab)^r$ can have two right contexts; among them, $(ab)^m$ has only the right context $c$. A substring containing $c$ can be left-maximal only when it is $(ab)^m c$ or $T$.

The right-extensions of these maximal repeats are as follows:
\[
\D_T(\varepsilon)=4,\qquad
\D_T((ab)^r)=2\quad (1\le r\le m-1),\qquad
\D_T((ab)^m c)=2.
\]
Hence
\[
\size(T)=4+2(m-1)+2=2m+4.
\]
Moreover, the nodes of in-degree one are exactly
\[
(ab)^2,(ab)^3,\ldots,(ab)^{m-1},
\]
and therefore
\[
\vone(T)=m-2.
\]
By Fact~2,
\[
\Gsize(T)=\size(T)-\vone(T)=(2m+4)-(m-2)=m+6.
\]
Equivalently, the corresponding reduced CDAWG grammar can be described as follows.
Let \(X_{\mathtt{ab}}\) be the nonterminal corresponding to the node \((ab)\),
let \(X_D\) be the nonterminal corresponding to the node
\(D=(ab)^m c\), and let \(X_T\) be the start nonterminal.
After removing the unary rules induced by the in-degree-one nodes
\((ab)^2,\ldots,(ab)^{m-1}\), the grammar has the productions
\[
X_{ab}\to ab, \qquad
X_D\to \underbrace{X_{ab}\cdots X_{ab}}_{m\ \mathrm{copies}}c,
\qquad
X_T\to X_D\,X_D\,\mathtt{\$}.
\]
The total length of the right-hand sides is \(2+(m+1)+3=m+6\), in agreement
with the formula above.

We next compute $\Gsize(T')$. The maximal substrings of $T'$, except for the sink $T'$, are precisely
\[
\varepsilon,\quad
(ab),(ab)^2,\ldots,(ab)^m,\quad
b,b(ab),b(ab)^2,\ldots,b(ab)^{m-1}.
\]
The substitution of the first $c$ by $b$ creates the unique factor $bb$ between the two alternating blocks. This makes the strings $b(ab)^r$ left- and right-maximal for $0\le r\le m-1$, while the strings $(ab)^r$ remain left- and right-maximal for $1\le r\le m$. All other factors either have a unique left context or a unique right context, or contain the unique end-marker and hence do not give additional maximal repeats.

Their right-degrees are
\[
\D_{T'}(\varepsilon)=4,
\]
\[
\D_{T'}((ab)^m)=2,
\]
and every other maximal repeat among the above $2m-1$ strings has three right-extensions. Thus
\[
\size(T')=4+3(2m-1)+2=6m+3.
\]
Furthermore, all the following $2m$ nodes have in-degree one:
\[
(ab),(ab)^2,\ldots,(ab)^m,\quad
b,b(ab),b(ab)^2,\ldots,b(ab)^{m-1}.
\]
Hence
\[
\vone(T')=2m.
\]
Again by Fact~2,
\[
\Gsize(T')=\size(T')-\vone(T')=(6m+3)-2m=4m+3.
\]
In terms of the reduced CDAWG grammar, the situation is particularly simple:
all internal maximal-repeat nodes of \(\CDAWG(T')\) have in-degree one, and hence
the unary rules corresponding to them are removed.  Thus the reduced grammar
has only the start production
\[
X_{T'}\to
\underbrace{a\,b\cdots a\,b}_{m \ \mathrm{copies \ of} \ ab}\,b\,
\underbrace{a\,b\cdots a\,b}_{m \ \mathrm{copies \ of} \ ab}\,
c\,\$.
\]
The length of this right-hand side is \(2m+1+2m+2=4m+3\), again agreeing
with \(\Gsize(T')\).

Consequently,
\[
\Gsize(T')-\Gsize(T)=(4m+3)-(m+6)=3m-3.
\]
Since $\Gsize(T)=m+6$, we obtain
\[
\Gsize(T')-\Gsize(T)=3\Gsize(T)-21.
\]
This proves the theorem.
\end{proof}
Fig.~\ref{fig:lowerbound_new} shows
the CDAWGs of the strings $T$ and $T'$ with $m = 3$.

\begin{figure}[H]
  \centering
  \includegraphics[width=0.7\linewidth]{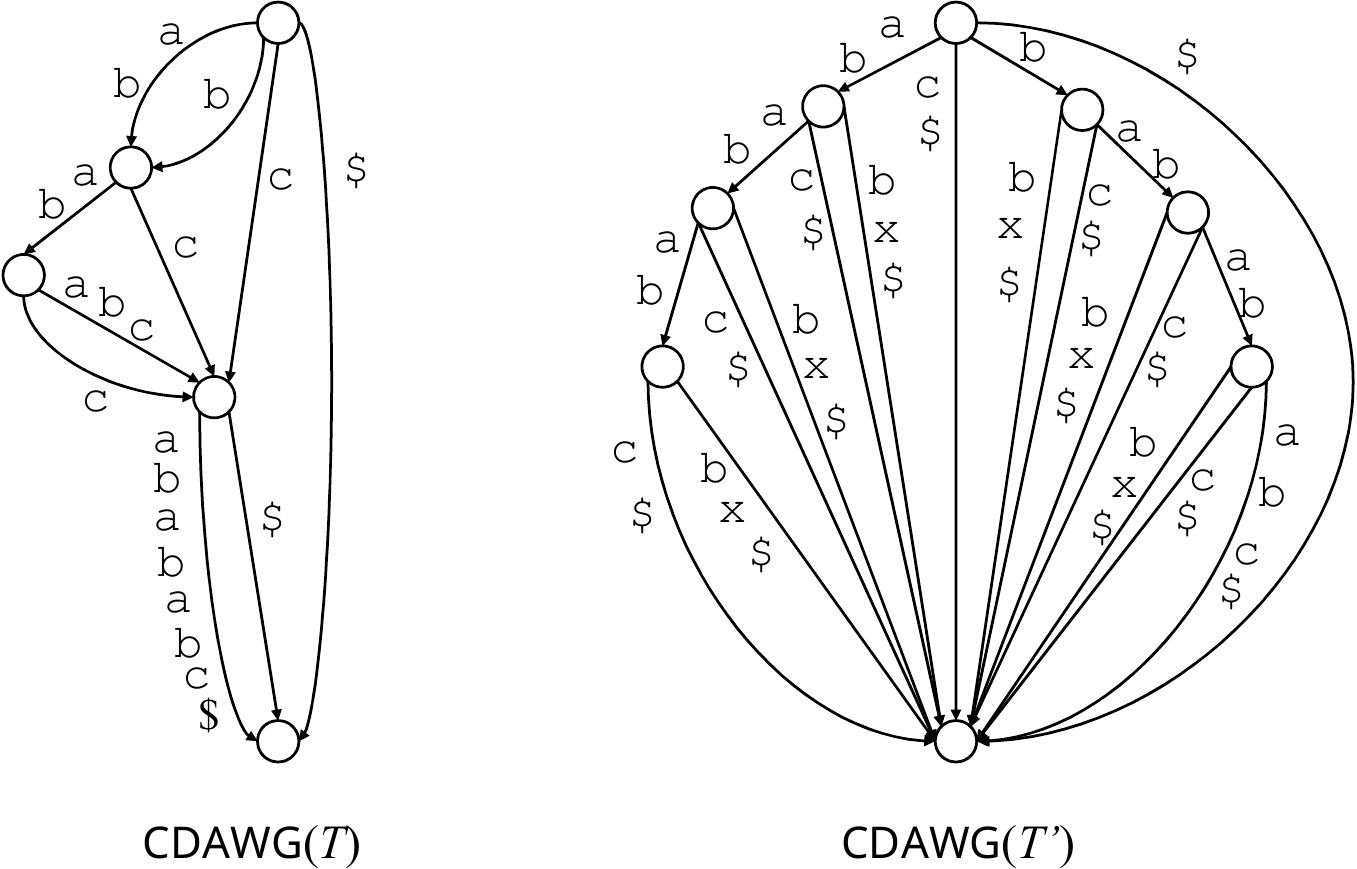}
  \caption{CDAWGs for the lower-bound example with $m=3$. Here
  $T=(\mathtt{ab})^3\mathtt{c}(\mathtt{ab})^3\mathtt{c}\$=\mathtt{abababcabababc}\$$ and
  $T'=(\mathtt{ab})^3\mathtt{b}(\mathtt{ab})^3\mathtt{c}\$=\mathtt{abababbabababc}\$ = \mathtt{abababbx}\$$ with $\mathtt{x} = \mathtt{abababc}$.
  For $T$, the reduced CDAWG grammar has productions
  $X_{\mathtt{ab}}\to \mathtt{a}\mathtt{b}$,
  $X_{\mathtt{x}}\to X_{\mathtt{ab}}X_{\mathtt{ab}}X_{\mathtt{ab}}\mathtt{c}$, and
  $X_T\to X_{\mathtt{x}}X_{\mathtt{x}}\mathtt{\$}$.
  For $T'$, all internal maximal-repeat nodes have in-degree one, and the reduced CDAWG grammar consists of
  $X_{T'}\to \mathtt{a}\mathtt{b}\mathtt{a}\mathtt{b}\mathtt{a}\mathtt{b}\mathtt{b}\mathtt{a}\mathtt{b}\mathtt{a}\mathtt{b}\mathtt{a}\mathtt{b}\mathtt{c}\mathtt{\$}$.
  We have $\size(T)=10$, $\vone(T)=1$, and $\Gsize(T)=9$, while
  $\size(T')=21$, $\vone(T')=6$, and $\Gsize(T')=15$.
  Hence $\Gsize(T')-\Gsize(T)=6=3\Gsize(T)-21$.}
  \label{fig:lowerbound_new}
\end{figure}

\section{Edge-sensitivity ingredients for CDAWG-grammars}
\label{sec:edges}

In this section, we revisit the CDAWG sensitivity analysis of Hamai et al.~\cite{HamaiFI2025}
and derive an upper bound on $\size(T')-\size(T)$ tailored to our analysis of
CDAWG-grammars.

Throughout this paper, $T$ denotes an arbitrarily fixed string
and $T'$ denotes the string obtained by performing a single-character edit operation on $T$. 
For an insertion or a substitution, let $i$ be the edit position in $T'$
that has the inserted or substituted character.
For a deletion, if the deleted character is $T[i]$ at position $i$ in $T$,
then the edit position in $T'$ is regarded as the empty gap between positions $i-1$ and $i$ in $T'$.
In the deletion case, $i$ also denotes the position immediately to the right of the gap in $T'$.

\begin{definition}
Let $x=T'[j..k]$ be a non-empty substring occurrence of $T'$,
where $1 \leq j \leq k \leq |T'|$.
This occurrence $[j..k]$ of $x$ in $T'$
is called \emph{crossing} if it contains or touches the edit
position $i$.  More precisely, it is crossing if
$k=i-1$ or $j\le i\le k$ or $j=i+1$ for an insertion or a substitution,
and if $j \le i \le k+1$ for a deletion.
\end{definition}

Let $x_L$ (resp. $x_R$) denote the leftmost (resp. rightmost) crossing occurrence of $x$ in $T'$,
and let $j_L$ and $k_L = j_L + |x|-1$ (resp. $j_R$ and $k_R = j_R + |x|-1$)
be the beginning and ending positions of $x_L$ (resp. $x_R$) in $T'$.
For $x_L$ and $x_R$, we consider the substrings
$S_{x_L} = T'[i..k_L]$ and $S_{x_R} = T'[i..k_R]$
of $T'$ (see Fig.~\ref{fig:x_Lp_L}),
which may be empty.

We divide $\MR(T')$ into the following four subsets:

\begin{figure}[H]
    \centering
  \includegraphics[keepaspectratio,scale=0.5]{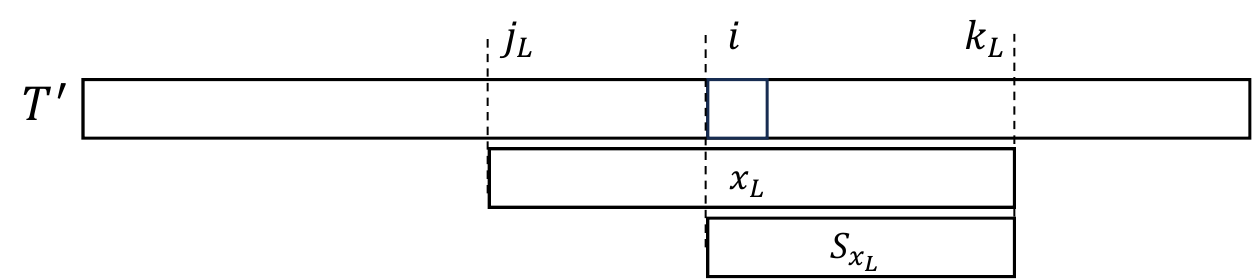}
  \caption{Illustration of $x_L$ in $T'$ for the case where $x_L$ contains $i$.}
  \label{fig:x_Lp_L}
\end{figure}

\begin{definition}
    Let 
    $\N=(\M(T')\setminus \M(T)) \setminus \{T'\}$ and $\Q=\MR(T')\setminus \N$~(or equivalently $\Q$ = $\MR(T') \cap \MR(T)$).
    We divide $\N$ into $\None$, $\Ntwo$, and $\Nthree$ as 
    $\None = \N\cap \RightM(T) $,
    $\Ntwo = \N\cap \LeftM(T) $,
    and $\Nthree = \N\backslash (\None \cup \Ntwo)$.
    Further, let $\Nadd$ denote the set of strings $x \in \Nthree$ satisfying the following two conditions:
    (1) If all occurrences of $x$ in $T'$ are crossing occurrences of $x$ in $T'$, then $x \in \MR(T')$, and
    (2) There is no distinct right-extension of $x$ in $T'$ other than the right-extension(s) of the crossing occurrence(s) of $x$.
    Let $\Nbase = \Nthree \setminus \Nadd$.
\end{definition}

Also, we define the number of remaining edges after edits by
$\remedge(T,T') = |\mathsf{E}_T \cap \mathsf{E}_{T'}|$.
We write it simply as $\remedge$ when it is clear from the context.

\subsection{Upper bound for total out-degrees of nodes w.r.t. $\None \cup \Nbase$}
Here, we show an upper bound for the total out-degrees of nodes corresponding to the strings that are elements of $\None \cup \Nbase$.

\begin{lemma}[Lemma 2 and 3 of~\cite{HamaiFI2025}]
  If $x \in \None \cup \Nbase$, there is no pair $(y, z) \subseteq \None \cup \Nbase$ of distinct strings ($y \neq z$) with $|x| < |y|$ and $|x| < |z|$
  such that
  both $S_{x_L} = S_{{y}_{F}}$ and $S_{x_R}=S_{{z}_{G}}$ hold at the same time, where $F,G \in \{L,R\}$.
\end{lemma}

From now on we consider the correspondence $\None \cup \Nbase$ and $\MR(T)$.
For any $x \in \None\cup \Nbase$, if there is no $y\in\None\cup\Nbase$ ($|y|>|x|, G\in \{L,R\}$) s.t.
$S_{x_L} = S_{y_G}$, then we associate $x$ with $S_{x_L}$, and otherwise we associate $x$ with $S_{x_R}$.

Let $U(x)$ be the string that is obtained by extending $S_{x_G}$, to which $x$ corresponds, 
to the left in $T$ until becoming left-maximal in $T$.
The following holds:

\begin{lemma}[Lemmas 4 and 5 of~\cite{HamaiFI2025}]
    \label{lem:N1edge}
    If $x,y \in \None \cup \Nbase$ and $x \neq y$, then $U(x),U(y)\in \M(T)$ and $U(x) \neq U(y)$.
    Also, $\D_{T'}(x) \leq \D_{T}(U(x)) + 2$.
\end{lemma}

Also, for each $x\in N_1\cup N_{3A}$, we define $\Ndelta(x)=2-\max(d_{T'}(x)-d_T(U(x)),0)$.
Then, we have the following:
$\sum_{x \in \None \cup \Nbase}\D_{T'}(x) 
\leq \sum_{x \in \None \cup \Nbase}(\D_T(U(x)) + 2-\Ndelta(x))
\leq \size(T) + 2|\M(T)| - \sum_{x \in \None \cup \Nbase}(\Ndelta(x))$.



\begin{lemma}
    \label{lem:dt1}
    $\sum_{x \in \None \cup \Nbase}\D_{T'}(x) \leq \size(T) + 2|\M(T)| - \sum_{x \in \None \cup \Nbase}\Ndelta(x)$.
\end{lemma}

\subsection{Upper bound for total out-degrees of nodes w.r.t. $\Ntwo \cup \Q$}
Here, we consider an upper bound for the total out-degrees of nodes corresponding to the strings that are elements of $\Ntwo \cup \Q$.
We divide the total out-edges of the nodes into two:
the number $\uc(T)$ of \emph{the right-extensions of the non-crossing occurrences} 
and the number $\wc(T)$ of \emph{the right-extensions only in the crossing occurrences},
  such that $\sum_{x \in \Ntwo\cup\Q}\D_{T'}(x) = \uc(T) + \wc(T)$.

For $\uc(T)$, the right-extensions of $x$ with non-crossing occurrences in $T'$ does not exceed the right-extensions of $x$ in $T$ 
because the non-crossing occurrences are not affected by the edit.
If $x\in \Ntwo$, then $x$ is not right-maximal in $T$. Thus, each node has at most one right-extension with non-crossing occurrences.
\hfnote*{modified}{%
On the other hand, each remaining right extension of $x\in \Q$ corresponds to the element of $E(T)\cap E(T')$ by definition.
By summing up these, we get:

\begin{lemma}
    $\uc(T) \leq |\Ntwo| + \remedge$.
\end{lemma}
}%
From now on we consider the second term $\wc(T)$.
A right-extension $a \in \Sigma$ of $x \in \Ntwo \cup \Q$
is said to be a \emph{new} right-extension of $x$ in $T'$,
if $xa$ does not occur in $T$
and $x$ has a crossing occurrence that immediately precedes $a$ in $T'$.

\begin{lemma}[Adapted from~\cite{HamaiFI2025}]
\label{lem:size_L}
There is an injection  from the
new right-extensions of all strings in $\Ntwo \cup \Q$
to the maximal repeats in $\M(T)$ (namely to the existing nodes in $\CDAWG(T)$),
implying $\wc(T) \leq |\M(T)|$.
\end{lemma}


Putting these results together, the following lemma holds:

\begin{lemma}
    \label{lem:dt2}
    $\sum_{x \in \Ntwo\cup\Q}\D_{T'}(x) \leq \remedge + |\M(T)| + |\Ntwo|$.
\end{lemma}

The following upper bound is known for $\Nadd$:
\begin{lemma}[Reformulation of Lemma 15 of~\cite{HamaiFI2025}]
    \label{lem:dt3}
    $\sum_{x \in \Nadd}\D_{T'}(x) \leq 2|\MR(T)|$.
\end{lemma}

Recall $|\MR(T)| = |\M(T)\setminus \{T\}| = |\M(T)|-1$.
By Lemmas~\ref{lem:dt1}, \ref{lem:dt2}, and~\ref{lem:dt3}
we get:
\begin{theorem}
    \label{lem:arb_edges}
$\size(T')-\size(T) \leq \remedge + 5|\MR(T)| - \sum_{x \in \None \cup \Nbase}(\Ndelta(x)) + |\Ntwo| + 3$.
\end{theorem}

\section{Tighter bound for $\AS(\Gsize, n)$}


This section analyzes the second half
  $\vone(T) - \vone(T')$
of Equation~\ref{eqn:grammar_size}
between the original string $T$ and the edited string $T'$.
By combining it with the result from Section~\ref{sec:edges},
we will obtain a tighter bound $\AS(\Gsize, n) \leq 8\Gsize + 4$ (Corollary~\ref{coro:main}).

\noindent \textbf{Difference in the number of nodes of in-degree one:}
We consider the number of $\G(T)$ \emph{existing nodes whose in-degrees are changed from one after the edit}.
Clearly $\vone(T) - \vone(T') \leq \G(T)$ holds.

Let $\I_T(x) = \{y \mid (y,x) \in \mathsf{E}_T\}$,
where $\mathsf{E}_T$ is the set of edges of $\CDAWG(T)$.
Let $x\in \MR(T)$ be a node of in-degree one
in the CDAWG for the original string $T$ before the edit.
We consider the following two cases:
\begin{description}
    \item[Case 1:] $x$ is not a maximal repeat in $T'$ after the edit (the node corresponding to $x$ is deleted in $\CDAWG(T')$).
    \item[Case 2:] $x$ has two or more in-edges in $\CDAWG(T')$ after the edit.
\end{description}
Let $\VO$ and $\mathsf{V}_2$ denote the sets of strings $x$ in Cases 1 and 2, respectively.

\begin{figure}[t]
    \centering
    \includegraphics[keepaspectratio,scale=0.45]{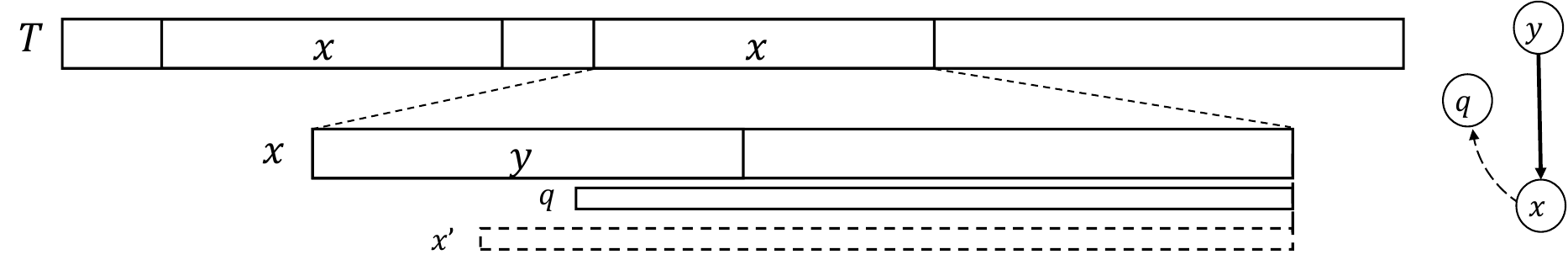}
    \caption{Illustration for the strings $y$ and $q$ for a string $x \in \mathsf{V}_2(T)$ and the relation of their nodes, where the dashed arc represents the suffix link. Any proper suffix $x'$ of $x$ that is longer than $q$ is not a maximal repeat of $T$.}
    \label{fig:arb_yq}
\end{figure}

In Case 2, we further define $y$ as the (only) parent of a node $x \in \mathsf{V}_2$
in the CDAWG of the original string $T$.
Also, let $q$ be the longest proper suffix of $x$ that is a maximal repeat in $T$.
See Fig.~\ref{fig:arb_yq}.
We categorize strings $x \in \mathsf{V}_2$ using the unique parent $y$ and its suffix $q$,
as follows:

\begin{description}
    \item[Case 2(a):] $y\notin \MR(T')$ (the node for $y$ is deleted in $\CDAWG(T')$).
    \item[Case 2(b):] $y\in \MR(T')$ and $q \notin \MR(T')$ (the node for $q$ is deleted in $\CDAWG(T')$).
    \item[Case 2(c):] $y,q \in \MR(T')$ (both of the nodes for $y$ and $q$ exist in $\CDAWG(T')$).
\end{description}

Let $\VA$, $\VB$, and $\VC$ denote the sets of strings $x \in \MR(T)$ which are in Case 2(a), Case 2(b), and Case 2(c), respectively.

\subsection{Cases 1, 2(a), 2(b)}
For each $x\in \VO\cup\VA\cup\VB$, we define $f(x) \in \MR(T)$
as follows:
(1) $f(x) = x$ if $x\in \VO$,
(2) $f(x) = y$ if $x\in \VA$, and
(3) $f(x) = q$ if $x\in \VB$.
Also, let $f(V) = \{f(x) \mid x\in V\}$ for each $V \in \{\VO, \VA, \VB\}$.

By definitions, the node for $f(x)$ and its out-edges are deleted in $\CDAWG(T')$.
Since we have assumed that $T$ terminates with a unique end-marker $\$$,
$\D_T(f(x)) \geq 2$.
Thus, at least $2|f(\VO)\cup f(\VA) \cup f(\VB)|$ edges are removed.

\hfnote*{added}{%
By the definition of $\remedge$, we have that:
\begin{lemma}
    \label{lem:size_eprime}
    At least $|\VO| + |\VA| + |\VB|$ existing edges are removed after edits. Then, $\size(T) \geq \remedge + (|\VO| + |\VA| + |\VB|)$.
\end{lemma}
}%

For $f(x)$, we have the following lemmas.

\begin{lemma}
    \label{lem:y_injective}
    If $x_1, x_2 \in \VA$ and $x_1 \neq x_2$, then $y_1 \neq y_2$.
\end{lemma}

\begin{proof}
    Suppose that $y_1 = y_2$. Then $y_1$ is the longest proper prefix of both $x_1$ and $x_2$ that is a maximal repeat in $T$,
    which implies that $a \neq b$ where $a,b \in \Sigma$, $y_1a$ is a prefix of $x_1$, and $y_1b$ is a prefix of $x_2$.
    Since $x_1, x_2$ are substrings of $T'$ after the edit, $y_1 \in \M(T')$, which is a contradiction.
\end{proof}

By symmetry, the following lemma also holds.

\begin{lemma}
    \label{lem:q_injective}
    If $x_1, x_2 \in \VB$ and $x_1 \neq x_2$, then $q_1 \neq q_2$.
\end{lemma}

\begin{lemma}
    \label{lem:yq_injective}
    If $x_1 \in \VA$ and $x_2 \in \VB$, then $f(x_1) \neq f(x_2)$.
\end{lemma}

\begin{proof}
    Suppose that $f(x_1) = f(x_2)$.
    By definition, $f(x_1)$ is a prefix of $x_1$ and a suffix of $x_2$.
    Since $x_1, x_2\in \M(T')$, $f(x_1) \in \M(T')$, which is a contradiction. 
\end{proof}

By Lemmas~\ref{lem:y_injective} and \ref{lem:q_injective}, there is an injective mapping between $V$ and $f(V)$ for each $V \in \{\VO, \VA, \VB\}$.
Thus, $|V| = |f(V)|$ holds for each $V$.
Also, for each element $x$, $x$ can be counted twice in $f(\VO)$ and $f(\VA)$ or $f(\VB)$,
while there is no $x \in f(\VA) \cap f(\VB)$ by Lemma~\ref{lem:yq_injective}.

As a result, $|\VO|+|\VA|+|\VB| \leq 2|f(\VO)\cup f(\VA) \cup f(\VB)|$ holds (see Fig.~\ref{fig:func}).
The number of removed edges is at least the quantity on the right-hand side,
which implies that at least $|\VO|+|\VA|+|\VB|$ edges are removed.

\begin{figure}[t]
    \centering
    \includegraphics[keepaspectratio,scale=0.50]{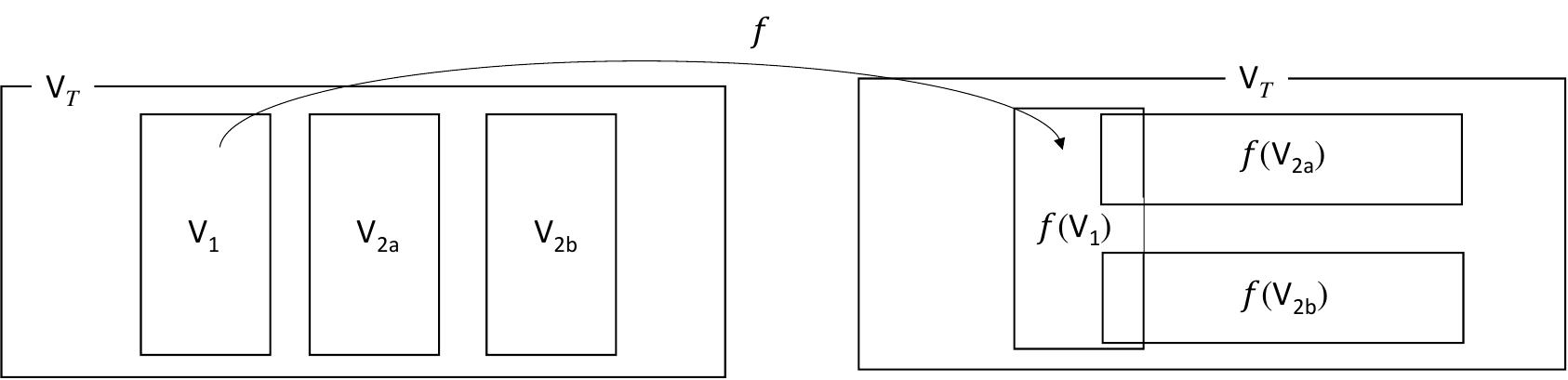}
    \caption{Illustration for $V$ and $f(V)$. the size of $V$ is the same as the size of $f(V)$ for each $V$.}
    \label{fig:func}
\end{figure}

\subsection{Case 2(c)}
We now consider the case $x\in \VC$. If the in-degree of $x$ is at least 2 in $\CDAWG(T')$, 
the following lemma holds:

\begin{lemma}
    \label{lem:x_inedge}

    If $y_1,y_2 \in \I_{T'}(x)$, at least one of $y_1, y_2$ has a non-prefix occurrence in $x$.
\end{lemma}

\begin{proof}
  Assume, to the contrary, that each of $y_1$ and $y_2$ occurs 
  exactly once in $x$, as a prefix of $x$ (see Fig.~\ref{fig:x_prefix}).
  Without loss of generality, let $|y_1| < |y_2|$.
  Let $a$ be the character such that $y_1a$ is a prefix of $x$.
  Since $y_1$ and $y_2$ are prefixes of $x$,
  the right-representative $\rrep_{T'}(y_1a)$ of $y_1a$ is not longer than $y_2$
  and thus it is not $x$.
  If there is no other occurrence of $y_1$ in $x$, then $(y_1, x) \notin \mathsf{E}_{T'}$, but this is a contradiction.
    Thus, at least one of $y_1, y_2$ has a non-prefix occurrence in $x$.
\end{proof}

By Lemma~\ref{lem:x_inedge}, if $\I_{T'}(x)$ has two elements, then at least one substring of $x$ has a non-prefix occurrence in $x$. 
Now, for $k \geq 0$, we define $z_k$ and $P_{z_k}$ as follows:
    \begin{itemize}
       \item $z_0$ is the maximal repeat in $\I_{T'}(x)$ whose ending position within $x$, denoted $j_{z_0}^{(x)}$, is the rightmost among those in $\I_{T'}(x)$.
        \item $P_{z_k}$ is the maximal repeat which is the longest prefix of $z_k$ where $P_{z_k} \in \I_{T'}(z_k)$.
        \item $z_{k+1}$ is the maximal repeat in $\I_{T'}(z_k)$ whose ending position within $z_k$ in $x$, denoted $j_{z_{k+1}}^{(x)}$, is the rightmost among those in $\I_{T'}(z_k)$.
    \end{itemize}
If there exist two or more maximal repeats with the same rightmost ending position, then the shortest one is selected as $z_k$.
    $P_{z_k}$ and $z_{k+1}$ are undefined when $z_k$ is a suffix of $y$ or $z_k$ is in-degree 1.  See also Fig.~\ref{fig:z_Pz}.

\begin{figure}[t]
    \centering
    \includegraphics[keepaspectratio,scale=0.55]{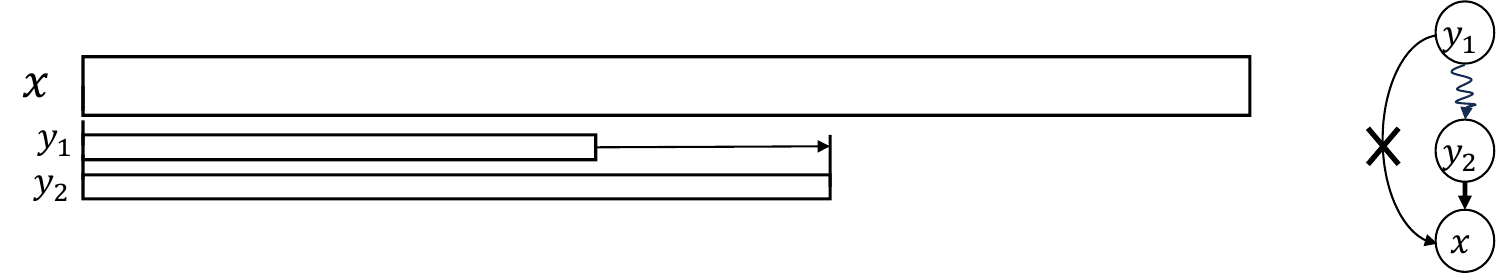}
    \caption{Illustration for the proof of Lemma~\ref{lem:x_inedge}. The case does not occur.}
    \label{fig:x_prefix}
\end{figure}




\begin{figure}[t]
    \centering
    \includegraphics[keepaspectratio,scale=0.48]{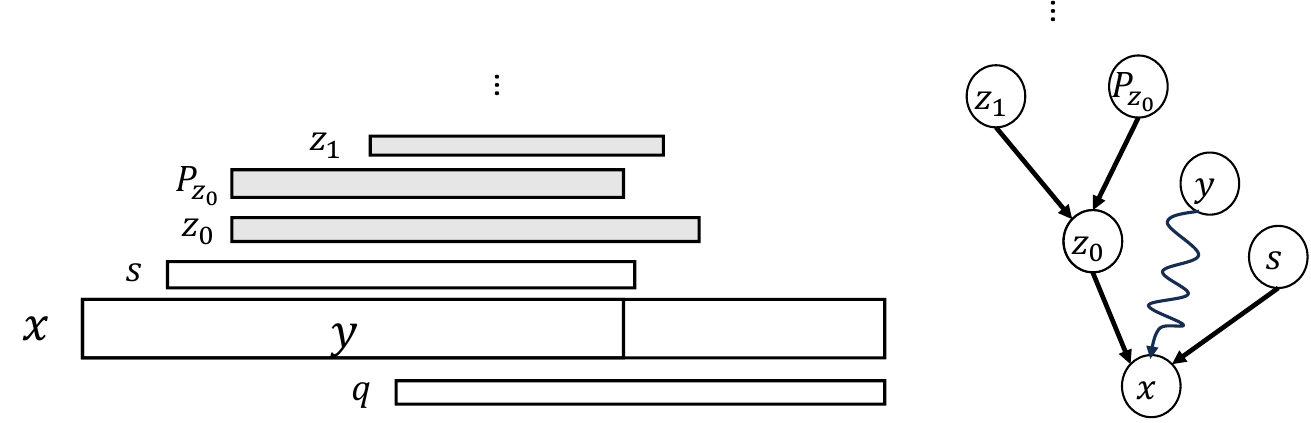}
    \caption{Illustration for $z_k$ and $P_{z_k}$ in $x$ and $\CDAWG(T')$. 
    Both $s$ and $z_0$ are elements of $I_{T'}(x)$, then $z_0$ is defined whose occurrence ends at the rightmost position.}
    \label{fig:z_Pz}
\end{figure}

For any $k$, let $i_{z_k}^{(x)}$ be the beginning position of $z_k$ in $x$.
This definition implies that $x$ has in-edges from $y$ and $z_0$, and $z_k$ has in-edges from $z_{k+1}$ and $P_{z_k}$ in $\CDAWG(T')$.
Since $(z_{k+1}, z_k) \in \mathsf{E}_{T'}$, $|z_k| > |z_{k+1}|$ holds. 
Therefore, the recursive construction of $P_{z_k}$ and $z_{k+1}$ stops at some $z_k$.

Let $i_y^{(x)}$ and $j_y^{(x)}$ (resp. $i_q^{(x)}$ and $j_q^{(x)}$) denote the beginning and ending positions of $y$ (resp. $q$) as prefix (resp. suffix) of $x$.
For $q$ and $z_k$, we have:

\begin{lemma}
    \label{lem:position_of_z}
    For $k \geq 0$, $i_{z_k}^{(x)} < i_q^{(x)} < j_y^{(x)} \leq j_{z_k}^{(x)}$ holds.
\end{lemma}

\begin{figure}[t]
    \centering
    \includegraphics[keepaspectratio,scale=0.54]{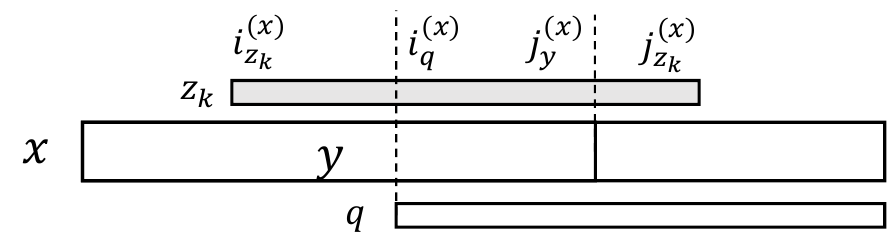}
    \caption{Illustration for Lemma~\ref{lem:position_of_z}.}
    \label{fig:position_of_z}
\end{figure}

\begin{proof}
  We prove the lemma by induction
  (see also Fig.~\ref{fig:position_of_z}).

  Let $z'_0 = x[j_{z_0}^{(x)}+1 .. |x|]$. Then $z_0z'_0$ is a suffix of $x$, and $z'_0$ is the label of the edge $(z_0, x)$ in $\CDAWG(T')$.
        Assume on the contrary that $i_{z_0}^{(x)} \geq i_q^{(x)}$. Then $z_0z'_0$ is also a suffix of $q$ and $z'_0$ is the label of the edge $(z_0, q)$, which is a contradiction.
        Assume on the contrary that $j_y^{(x)} > j_{z_0}^{(x)}$. Then $z_0$ is a substring of $y$. Since $y\in \MR(T')$, every path from $z_0$ to $x$ goes through the node $y$.
        This implies that $z_0$ does not have an out-edge to $x$, a contradiction.
        By the definition of $q$, any suffix of $x$ that is longer than $q$ is represented
        by the same node as $x$.
        Since $x$ has in-degree one in $\CDAWG(T)$ and its unique in-neighbor is $y$,
        $y$ and $q$ must overlap in $x$. Hence $i_q^{(x)}< j_y^{(x)}$.      
        Thus, $i_{z_0}^{(x)} < i_q^{(x)}< j_y^{(x)} \leq j_{z_0}^{(x)}$.

        If $i_{z_k}^{(x)} < i_q^{(x)} < j_y^{(x)} \leq j_{z_k}^{(x)}$ holds and $z_{k+1}$ is generated, then $j_y^{(x)} < j_{z_k}^{(x)}$ holds by definition.
        Assume on the contrary that $i_{z_{k+1}}^{(x)} \geq i_q^{(x)}$. Let $s_{z_k}$ denote the overlap between $z_k$ and $q$. Then $z_{k+1}$ is a substring of $s_{z_k}$. 
        Since $z_k, q \in \MR(T')$, $s_{z_k} \in \MR(T')$.
        Thus, every path from $z_{k+1}$ to $z_k$ goes through the node $s_{z_k}$
        (and thus $z_{k+1}$ does not have an out-edge to $z_k$), a contradiction.
        Assume on the contrary that $j_y^{(x)} > j_{z_{k+1}}^{(x)}$. Then $z_{k+1}$ is a substring of $y$. Since $y\in \MR(T')$, every path from $z_{k+1}$ to $x$ goes through the node $y$.
        This implies that $z_{k+1}$ does not have an out-edge to $z_k$, which is a contradiction.
        Thus, $i_{z_{k+1}}^{(x)} < i_q^{(x)}< j_y^{(x)} \leq j_{z_{k+1}}^{(x)}$.
\end{proof}

Consider the string $z_k$ such that $z_{k+1}$ is not generated. Let $\zl = z_k$.
In the sequel, we consider two sub-cases of Case 2(c):


\subsection{Case 2(c)-1: $\zl$ has only a single in-edge and it is not a suffix of $y$}
Let $\Vbt(T)$ denote the set of strings $x$ in Case 2(c)-1.
We remark that $\zl$ has only a single in-edge, and $\zl$ is a new node whose in-degree is one in $\CDAWG(T')$, instead of $x$.
Also, we have the following lemma:

\begin{lemma}
    \label{lem:\zl_injective}
    If $x_1 \neq x_2$ and $z_l(x_1), z_l(x_2)$ are in Case 2(c)-1, then $z_l(x_1) \neq z_l(x_2)$.
\end{lemma}

\begin{proof}
    Suppose that $z_l(x_1) = z_l(x_2)$. 
    Let $x'$ be the maximal common substring of $x_1$ and $x_2$ which contains $z_l(x_1) = z_l(x_2)$.
    By the maximality of $x'$ as a common substring of the two occurrences in
    $T$, $x'$ is a maximal repeat in $T$.
    Also, let $y_1$ (resp. $q_1$) denote the longest proper prefix (resp. suffix) of $x_1$ that is in $\MR(T)$.
    Since $i_{z_l(x_1)}^{(x_1)} < i_{q_1}^{(x_1)} < j_{y_1}^{(x_1)} \leq j_{z_l(x_1)}^{(x_1)}$ by Lemma~\ref{lem:position_of_z} and $z_l(x_1)$ is not a suffix of $y_1$,
    $i_{x'}^{(x_1)} < i_{q_1}^{(x_1)} < j_{y_1}^{(x_1)} < j_{x'}^{(x_1)}$ also holds.
    This also implies that $x'\ne y_1$.
    Since $y_1$ is the only node that has an out-edge to $x_1$ in
$\CDAWG(T)$, and $q_1$ is the longest maximal repeat of $T$
that is a proper suffix of $x_1$, we derive a contradiction as follows:

If the considered occurrence of $x'$ is a suffix of $x_1$, then $x'$
itself is a suffix of $x_1$ longer than $q_1$, contradicting the definition
of $q_1$. Otherwise, let $a\in\Sigma$ be the character immediately
following this occurrence of $x'$ in $x_1$, so that
$x'a\in \Substr(x_1)$. Then $\rrep_T(x'a)$ is a suffix
of $x_1$ that is longer than $q_1$. If $\rrep_T(x'a)\ne x_1$,
then this contradicts the definition of $q_1$. Otherwise,
$\rrep_T(x'a)=x_1$, and hence $x'$ has an outgoing edge to
$x_1$ in $\CDAWG(T)$, contradicting that $y_1$ is the unique in-neighbor of $x_1$ in $\CDAWG(T)$.
\end{proof}
By Lemma~\ref{lem:\zl_injective}, there is an injection that maps a new node $\zl$ in $\CDAWG(T')$ to an existing node $x$ in $\CDAWG(T)$, 
and these nodes have only a single in-edge.
Thus, the number of nodes whose in-degree is one does not change in Case 2(c)-1. 

\subsection{Case 2(c)-2: $\zl$ is a suffix of $y$}
Let $\Val$ be the set of strings $x$ in Case 2(c)-2, and $\Valy$ be the set of strings $y$ corresponding to $x \in \Val$.
If $x\in \Val$, $\zl$ is right-maximal and is not left-maximal in $T$ since $\zl$ is a suffix of $y$ and $\zl\notin \MR(T)$.
Then, $\zl \in \None$.
Let $P_q$ denote the overlap of $y$ and $q$. Since $y, q\in \MR(T)$, $P_q\in\MR(T)$.
For $P_q$, we have:

\begin{lemma}
    \label{lem:P_q}
    $P_q$ has also a crossing-occurrence in $T'$. Furthermore, $U(\zl)$ (the existing maximal repeat corresponding to $\zl$) is a suffix of $P_q$.
\end{lemma}

\begin{proof}
    Suppose $P_q$ has only non-crossing occurrence(s) in $T'$.
    See Fig.~\ref{fig:P_q_cross}.
    $\zl$ has a crossing occurrence in $T'$ since $\zl \in \None$.
    Since $P_q$ has only non-crossing occurrence(s) in $T'$,
    $U(\zl)$ is a maximal repeat in $T$ which has $P_q$ as its proper suffix.
    Let $S_q$ be the string such that $q = P_qS_q$. Since $U(\zl)S_q$ is a suffix of $x$ longer than $q$, $U(\zl)S_q \EqrL_T x$ holds.
    Since $U(\zl)$ is also a suffix of $y$, there exists a path from $U(\zl)$ to $x$ which does not contain $y$ in $\CDAWG(T)$.
    This implies that $x$ has another in-edge from some node other than $y$ in $\CDAWG(T)$ before edit, which is a contradiction.
    By the above arguments, if $U(\zl)$ is longer than $P_q$, then it also leads to a contradiction.
    Since $U(\zl)$ and $P_q$ are suffixes of $y$, $U(\zl)$ is also a suffix of $P_q$.
\end{proof}

\begin{figure}[t]
    \centering
    \includegraphics[keepaspectratio,scale=0.4]{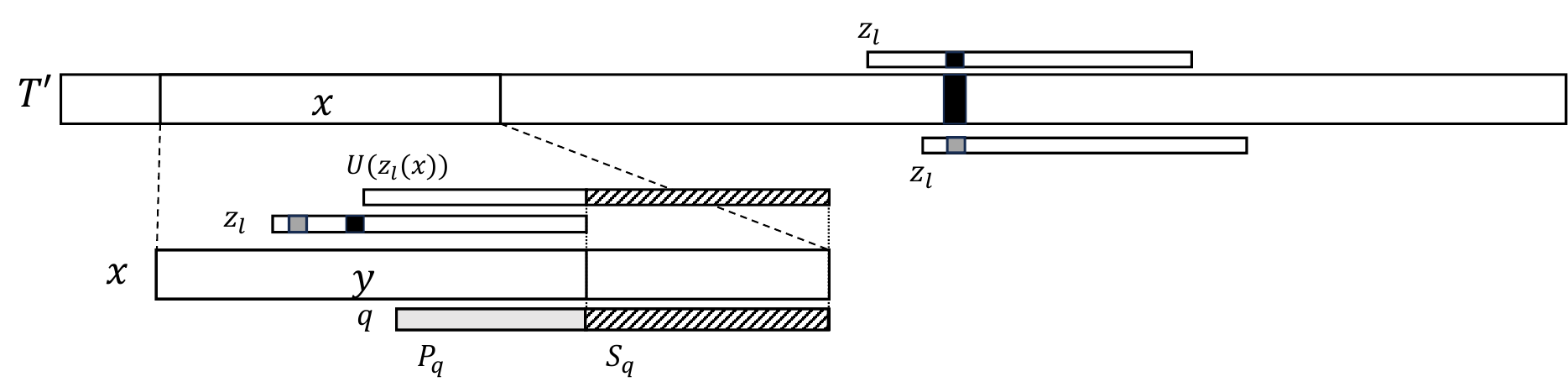}
    \caption{Illustration for the proof of Lemma~\ref{lem:P_q}. The black position in $T'$ is the edited position. 
    $P_q$ has a non-crossing occurrence which does not contain nor touch the edited position.}
    \label{fig:P_q_cross}
\end{figure}

Let $P_q(x)$ be a corresponding $P_q$ to $x\in\Val$.
For the correspondence between $y$ and $\zl, P_q$, we have the following lemma:



\begin{lemma}
    \label{lem:injective_y}
    For each $y_1, y_2\in V^{(y)}_{2(c)-2}$, choose an arbitrary corresponding $x_1, x_2\in V_{2(c)-2}$.  
    If $y_1 \neq y_2$, then $z_l(x_1) \neq z_l(x_2)$ and $P_q(x_1) \neq P_q(x_2)$.
\end{lemma}

\begin{proof}
    Assume $z_l(x_1) = z_l(x_2)$.
    Since $z_l(x_1)$ occurs both in $y_1$ and in $y_2$ as their suffix, there exists a maximal repeat of $T$ in $z_l(x_1)$ or in a suffix of $y_1$ that is longer than $z_l(x_1)$. 
    This implies that there exists a maximal repeat of $T$ that is longer than $P_q(x_1)$ (similar to Lemma~\ref{lem:P_q}), which is a contradiction.
    Thus, $z_l(x_1) \neq z_l(x_2)$.
    Assume $P_q(x_1)=P_q(x_2)$. Since this string occurs as a suffix of both $y_1$ and $y_2$, the same argument as above implies that there is a maximal repeat of $T$ contained in $P_q(x_1)$ or in a suffix of $y_1$
    that is longer than $P_q(x_1)$, which is a contradiction.
\end{proof}

By Lemma~\ref{lem:injective_y}, there exists an injective mapping from $y$ to $\zl$ and $P_q$. 

We further show the value of $|\Valy|$ by using $\zl$.
Since $\zl$ is a suffix of $y\in\MR(T)$, $\zl \in \None$, which implies $\D_{T'}(\zl) \leq \D_T(U(\zl)) + 2$ due to Lemma~\ref{lem:N1edge}.



Let $\Valz$ denote the set of $\zl$ where $x\in \Val$.
Let $\Vdelta(\zl)$ denote the value such that $\Vdelta(\zl) = \max(\D_{T'}(\zl) -\D_T(U(\zl)), 0)$, 
and $x_y$ denote the arbitrary chosen $x$ corresponding to each $y \in \Valy$.
Also, the following lemma holds:



\begin{lemma}
    \label{lem:size_N2}
    $|\Ntwo| \leq |\M(T)| - \sum_{y\in \Valy}\Vdelta(\zly)$.
\end{lemma}

\begin{proof}
    For each $y\in \Valy$, let $q=q(x_y)$ and $P_q=P_q(x_y)$.
    Since $q,y\in Q$ and $P_q$ is their overlap in the occurrence inside
    $x_y$, we have that $P_q$ is left-maximal and right-maximal in both $T$ and $T'$.
    Hence, $P_q \in \Q$.
    Also, by Lemma~\ref{lem:P_q}, $U(z_l(x_y))$ is a suffix of $P_q(x_y)$, and by
    Lemma~\ref{lem:position_of_z} and the definition of Case 2(c)-2, $P_q(x_y)$ is a suffix of $z_l(x_y)$.
    
    Suppose $\Vdelta(\zly) = 2$.
    Since $U(\zly)$ is a suffix of $P_q(x_y)$ and $P_q(x_y)$ is a suffix of $\zly$, $\D_{T'}(P_q(x_y)) \geq \D_T(U(\zly)) + 2 \geq \D_T(P_q(x_y)) +2$ holds.
    This implies that $P_q(x_y)\in Q$ has at least two new right-extensions
in $T'$. By Lemma~\ref{lem:size_L}, these new right-extensions are injectively mapped
to two distinct maximal repeats in $\M(T)$.
    Similarly, for each $\zly$ where $\Vdelta(\zly) = 1$, $\D_{T'}(P_q(x_y)) \geq \D_T(P_q(x_y)) + 1$ holds,
    and $P_q(x_y)\in Q$ has at least one new right-extension in $T'$, which
is injectively mapped to a maximal repeat in $\M(T)$ by Lemma~\ref{lem:size_L}.
    By Lemma~\ref{lem:injective_y}, the strings $P_q(x_y)$ for distinct $y\in \Valy$ are pairwise distinct.
    Then, at least $\sum_{y\in \Valy}\Vdelta(\zly)$ new right-extensions of these strings $P_q(x_y)\in \Q$ are injectively
mapped to distinct maximal repeats in $\M(T)$.
    On the other hand, each string $s\in \Ntwo$ has at least one new
    right-extension in $T'$, since $s\in \M(T')$ but $s\notin \RightM(T)$.
    By Lemma~\ref{lem:size_L}, these new right-extensions are also injectively mapped to
    distinct maximal repeats in $\M(T)$.
    Applying the injection of Lemma~\ref{lem:size_L} to the union of the above new right-extensions of strings in $\Ntwo \cup \Q$, we obtain
    $|\Ntwo|+\sum_{y\in \Valy}\Vdelta(z_l(x_y)) \le |\M(T)|$.
\end{proof}

Also, $\zl \in \None$ and $\Ndelta(\zl) = 2-\max(\D_{T'}(\zl) -\D_T(U(\zl)), 0)$, we have that:

\begin{corollary}
\label{cor:size_delta}
For every $x\in \Val$, $\Ndelta(\zl)+\Vdelta(\zl)=2.$
\end{corollary}

Also, the chosen strings $\zly$ are pairwise distinct elements of $\Valz$ from Lemma~\ref{lem:injective_y}.
Hence, we also have the following:

\begin{corollary}
\label{cor:size_VY}
$|\None| \geq |\Valz| \geq |\Valy|$ holds.
\end{corollary}

To consider the size of $|\Val|$, let $\Sy = \sum_{y\in\Valy}(\D_T(y))$, and $\Xy$ be the total number of in-edges of $x\in\Val$ in $\CDAWG(T)$.
See the upper diagram (a) of Fig.~\ref{fig:ytox}.
Since $x\in\Val$ has only one in-edge $(y,x)$ where $y\in\Valy$ in $\CDAWG(T)$, $|\Val| = \Xy \leq \Sy$ holds.

\begin{figure}[t]
    \centering
    \includegraphics[keepaspectratio,scale=0.48]{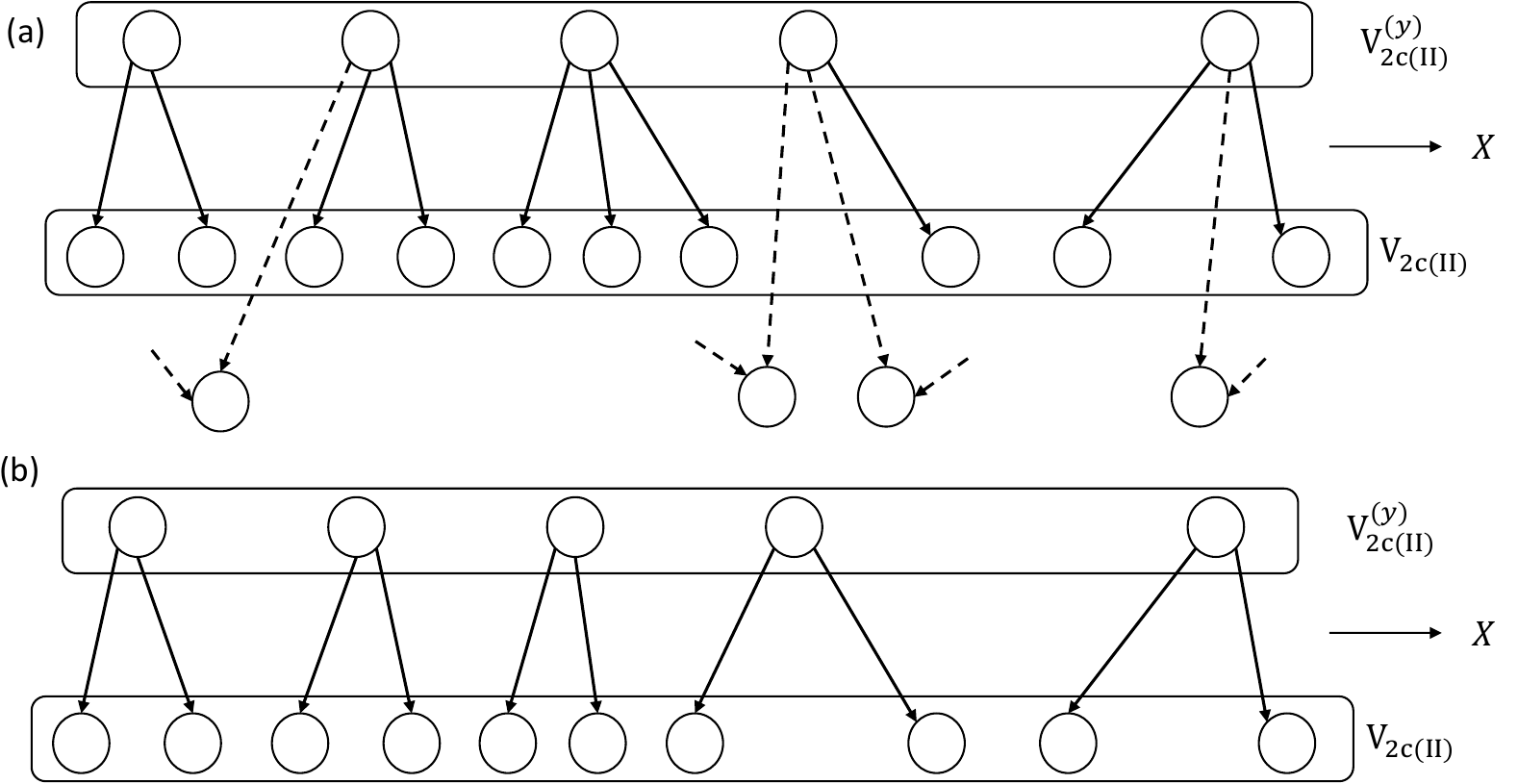}
    \caption{(a) Illustration for the definitions of $\Valy$, $\Xy$ and $\Sy$. $\Sy$ counts all of the out-edges from $y\in \Valy$,
    although $\Xy$ only counts the edges going to $x\in \Val$ (only solid-line edges are counted). 
    (b) The situation of (a) in which $\Gsize(T')-\Gsize(T)$ is maximized.}
    \label{fig:ytox}
\end{figure}

Now, we put all these results (Cases 1, 2(a), 2(b), 2(c)-1, and 2(c)-2) together.
Since $\G(T) = |\VO| + |\VA| + |\VB| + |\Vbt(T)| + |\Val|$ and
at least $|\Vbt(T)|$ new nodes have only one in-edge in Case 2(c)-1, 
we have the following:

\begin{lemma}
    \label{lem:arb_vone}
    $\vone(T) - \vone(T') \leq |\VO| + |\VA| + |\VB| + \Xy$ holds.
\end{lemma}


Finally, we combine the results of Theorem~\ref{lem:arb_edges}, Lemmas~\ref{lem:size_eprime},~\ref{lem:size_N2} and~\ref{lem:arb_vone}, and Corollaries~\ref{cor:size_delta} and~\ref{cor:size_VY}.
Then, we have:
\[
\begin{array}{l}
\Gsize(T') - \Gsize(T)\\
 =   (\size(T') - \size(T)) + (\vone(T) -  \vone(T'))\\
 \leq  \remedge + 5|\MR(T)| - \sum_{x \in \None \cup \Nbase}\Ndelta(x) + |\Ntwo| + 3 + |\VO| + |\VA| + |\VB| + \Xy\\
 \leq \size(T) + 6|\MR(T)| - \sum_{x \in \None \cup \Nbase}\Ndelta(x) - \sum_{y\in\Valy}\Vdelta(\zly) + \Xy + 4\\
\leq \size(T) + 6|\MR(T)| - \sum_{y \in \Valy}(\Ndelta(\zly) + \Vdelta(\zly)) + \Xy + 4\\
\leq \size(T) + 6|\MR(T)| + \Xy - 2|\Valy| + 4.
\end{array}
\]
We call this Inequality ($\alpha$).


Now we consider the value of $\Xy$. Let $\Snoty = \sum_{y\in\MR(T)\setminus\Valy}(\D_T(y))$.
For $\Sy$, $|\MR(T)|$ and $\size(T)$, the following holds:
\begin{lemma}
    \label{lem:value_MT}
    Let $\omega = \Xy - 2|\Valy|$, then $|\MR(T)| \leq \min(\size(T), \size(T)-\omega)/2$.
\end{lemma}

\begin{proof}
If $\omega \leq 0$, then $\size(T) \geq 2|\MR(T)|$ due to Fact~\ref{fact:outdegree}.
If $\omega > 0$, then $\Sy \geq \Xy = 2|\Valy| + \omega$ holds. Since $\Sy +\Snoty = \size(T)$ and $\Snoty \geq 2|\MR(T)\setminus\Valy|$,
$\size(T) \geq 2|\MR(T)| + \omega$. Then $|\MR(T)| \leq (\size(T)-\omega)/2$.
\end{proof}

By substituting the value of Lemma~\ref{lem:value_MT} to Inequality ($\alpha$) above,
we get $\Gsize(T') - \Gsize(T) \leq 4\size(T) + \omega + 4$ if $\omega \leq 0$,
or $\Gsize(T') - \Gsize(T) \leq 4\size(T) - 2\omega + 4$ if $\omega > 0$.
Therefore, $\Gsize(T') - \Gsize(T) \leq 4\size(T) + 4$. 
Fig.~\ref{fig:ytox}(b) illustrates the situation suggested by the above bound.
As a result, we have the following theorem:

\begin{theorem}\label{theo:main}
    The additive sensitivity of the CDAWG-grammar is at most $4\size + 4$.
\end{theorem}

Since $\Gsize(T) \geq \size(T)/2$, Theorem~\ref{theo:main} immediately leads to the following:

\begin{corollary}
    \label{coro:main}
    The additive sensitivity of the CDAWG-grammar is at most $8\Gsize + 4$.
\end{corollary}

\section{Conclusions}

In this paper, we investigated the additive sensitivity of CDAWG-grammars under a
single-character edit.  We showed that there exists a family of strings whose
CDAWG-grammar size increases by $3\Gsize-21$ after a single-character edit, yielding a lower
bound of $3\Gsize-21$.  We also proved that, for every string, CDAWG-grammar size can increase
at most $4\size + 4$ after any single-character edit.
Together with $\Gsize\ge \size/2$, this gives the upper bound of
$8\Gsize+4$.

The main remaining open problem is to make these bounds tighter. A more detailed
understanding of the change in in-degree-one nodes may be useful for closing this gap.

\bibliographystyle{abbrv}
\bibliography{ref}

\end{document}